\documentclass[12pt,reqno]{amsart}
\usepackage[notref,notcite]{}
\usepackage{amssymb}
\usepackage{amsmath}

\newtheorem{theorem}{Theorem}[section]
\newtheorem{lemma}[theorem]{Lemma}

\newtheorem{conjecture}[theorem]{Conjecture}
\numberwithin{equation}{section}

\def\D{\mathbb D}

\def\pa{\partial}

\def\Z{\mathbb Z}
\def\T{{\rm Tr}}

\begin{document}

\title{An inequality for the zeta function of a planar domain}
\author{Alexandre Jollivet and Vladimir Sharafutdinov}
\thanks{The work was started when the authors stayed at Institut Henri Poincar\'{e} in the scope of the trimester "Inverse Problems", May -- June 2015. The authors are grateful to the institute for the support and hospitality.\\
The first author is partially supported by French grant ANR-13-JS01-0006.\\
The second author was supported by RFBR, Grant 15-01-05929-a.}
\address{Laboratoire de Math\'ematiques Paul Painlev\'e,
CNRS UMR 8524/Universit\'e Lille 1 Sciences et Technologies,
59655 Villeneuve d'Ascq Cedex, France}
\email{alexandre.jollivet@math.univ-lille1.fr}

\address{Sobolev Institute of Mathematics and Novosibirsk State University, Russia}
\email{sharaf@math.nsc.ru}



\maketitle

\begin{abstract}
We consider the zeta function $\zeta_\Omega$ for the Dirichlet-to-Neumann operator of a simply connected planar domain $\Omega$ bounded by a smooth closed curve.
We prove non-negativeness and growth properties for $\zeta_\Omega(s)-2\big({L(\pa \Omega)\over 2\pi}\big)^s\zeta_R(s)\ (s\leq-1)$, where $L(\pa \Omega)$ is the length of the boundary curve and $\zeta_R$ stands for the classical Riemann zeta function.
Two analogs of these results are also provided.

\end{abstract}

\section{Introduction}

Let $\Omega$ be a simply connected (probably multisheet) planar domain bounded by a $C^\infty$-smooth closed curve $\partial\Omega$. See \cite[Section 4]{JS} for the discussion of simply connected multisheet planar domains. The {\it Dirichlet-to-Neumann operator} of the domain
$$
\Lambda_\Omega:C^\infty(\partial\Omega)\rightarrow C^\infty(\partial\Omega)
$$
is defined by $\Lambda_\Omega f=\frac{\partial u}{\partial\nu}|_{\partial\Omega}$, where $\nu$ is the outward unit normal to $\partial\Omega$ and $u$ is the solution to the Dirichlet problem
$$
\Delta u=0\quad\mbox{\rm in}\quad\Omega,\quad u|_{\partial\Omega}=f.
$$
The Dirichlet-to-Neumann operator is a first order pseudodifferential operator. Moreover, it is a non-negative self-adjoint operator with respect to the $L^2$-product
$$
(u,v)_{L^2(\partial\Omega)}=\int\limits_{\partial\Omega}u\bar v\,ds,
$$
where $ds$ is the Euclidean arc length of the curve $\partial\Omega$. In particular, the operator $\Lambda_\Omega$ has a non-negative discrete eigenvalue spectrum
$$
\mbox{\rm Sp}(\Omega)=\{0=\lambda_0(\Omega)<\lambda_1(\Omega)\leq\lambda_2(\Omega)\leq\dots\},
$$
where each eigenvalue is repeated according to its multiplicity. The spectrum is called the {\it Steklov spectrum} of the domain $\Omega$. In particular, for the unit disc
${\mathbb D}=\{(x,y)\mid x^2+y^2\leq1\}$,
\begin{equation}
\mbox{\rm Sp}({\mathbb D})=\{0=\lambda^0_0<\lambda^0_1\leq\lambda^0_2\leq\dots\}=\{0,1,1,2,2,\dots\}.
                                            \label{1.9}
\end{equation}

Steklov eigenvalues depend on the size of $\Omega$ in the obvious manner: $\lambda_k(c\Omega)=c^{-1}\lambda_k(\Omega)$ for $c>0$. The dependence of Steklov eigenvalues on the shape of $\Omega$ is probably the most interesting problem of the subject. Nevertheless, only a few of results are known in this direction for simply connected smooth domains. The earliest of such results is Weinstock's inequality \cite{We}:
\begin{equation}
\lambda_1(\Omega)\leq\frac{2\pi}{L(\partial\Omega)},
                                            \label{1.10}
\end{equation}
where $L(\partial\Omega)$ is the length of the boundary curve. The equality in (\ref{1.10}) holds if and only if $\Omega$ is a round disc. The first idea for generalizing Weinstock's inequality to higher Steklov eigenvalues is just to write
\begin{equation}
\lambda_k(\Omega)\leq\frac{2\pi}{L(\partial\Omega)}\,\lambda^0_k \quad(k=0,1,\dots),
                                            \label{1.11}
\end{equation}
where $\lambda^0_k$ are defined by (\ref{1.9}). However, inequality (\ref{1.11}) is wrong for $k\geq2$. The right version of such an estimate obtained by  Hersch, Payne and Schiffer \cite{HPS} (see also the recent work of Girouard and Polterovich \cite{GP2}) looks as follows:
\begin{equation}
\lambda_k(\Omega)\leq\frac{2\pi}{L(\partial\Omega)}\,k\quad (k=0,1,\dots).
                                            \label{1.12}
\end{equation}
The estimate is asymptotically sharp. Observe that, roughly speaking, inequality (\ref{1.12}) is twice weaker than (\ref{1.11}) since $\lambda^0_k\sim k/2$ for large $k$. Actually Hersch, Payne and Schiffer obtained the more general result
$$
\lambda_k(\Omega)\lambda_\ell(\Omega)\leq\left\{
\begin{array}{ll}
\Big(\frac{\textstyle \pi}{\textstyle L(\partial\Omega)}\Big)^2(k+\ell)^2,&\ \mbox{if}\ k+\ell\ \mbox{is even},\\
[0.2cm]
\Big(\frac{\textstyle \pi}{\textstyle L(\partial\Omega)}\Big)^2(k+\ell-1)^2,&\ \mbox{if}\ k+\ell\ \mbox{is odd}
\end{array}\right.
$$
that can be considered as a geometric inequality for pairs of Steklov eigenvalues.

Our main result, presented by Theorems \ref{Th1.1} and \ref{Th1.1'} below, can be also considered as a geometric inequality for the whole collection of Steklov eigenvalues.

Steklov eigenvalues are also defined in more general settings: for domains with non-smooth boundaries, for multiply connected domains, for multidimensional domains, and for compact Riemannian manifolds with boundaries. There are much more geometric inequalities obtained in such settings, see \cite[Section 4]{GP} and references there. In the present paper, we do not consider any of this settings because of the following specifics of the asymptotics of the Steklov spectrum which plays the crucial role in our approach.

In the case of a simply connected (probably multisheet) planar domain $\Omega$ bounded by a $C^\infty$-smooth closed curve, the asymptotics of the Steklov spectrum is completely determined by the length $L(\partial\Omega)$ of the boundary curve. More precisely,
\begin{equation}
\lambda_n(\Omega)=\frac{2\pi}{L(\partial\Omega)}\lambda^0_n+O(n^{-N})\quad\mbox{\rm as}\quad n\rightarrow\infty
                                            \label{1.13}
\end{equation}
fore any $N>0$.
To our knowledge, this fact was first proved by Rozenblum \cite{Ro}. Essentially the same proof was independently presented by Edward \cite{E} with the reference to some preprint by Guillemin and Melrose. In the most general setting, the proof is reproduced in \cite[Lemma 2.1]{GPPS}.

Asymptotics (\ref{1.13}) holds also for simply connected compact Riemannian surfaces with smooth boundary with the obvious change: $L(\partial\Omega)$ is the Riemannian length of the boundary. Studying such Riemannian surfaces is equivalent to studying Riemannian metrics on the unit disc ${\mathbb D}$. In this setting, our main result is expressed by Theorems \ref{Th1.3} and \ref{Th1.3'} below. The only circumstance should be taken into account: conformally equivalent Riemannian metrics have coincident Steklov spectra. The circumstance arises due to the conformal invariance of the Laplace -- Beltrami operator in the two-dimensional case.

Asymptotics (\ref{1.13}) can be generalized to the case  of multiply connected smooth planar domains as well as to the case of multiply connected compact Riemannian surfaces with boundary \cite{GPPS}. If the boundary has $m$ components with lengths $L(\partial_1\Omega),\dots,L(\partial_m\Omega)$, then every boundary component $\partial_i\Omega\ (i=1,\dots,m)$ gives the contribution (\ref{1.13}), with $L(\partial\Omega)$ replaced by $L(\partial_i\Omega)$, and the whole asymptotics is just the union of the contributions. Moreover, a solution to the boundary value problem
$$
\Delta u=0\ \ \mbox{in}\ \ \Omega,\quad\quad \frac{\partial u}{\partial\nu}=\lambda_iu\ \ \mbox{on}\ \ \partial_i\Omega\ \ (1\leq i\leq m)
$$
is concentrated near the boundary and solutions coming from different boundary components do not interfere for large frequencies $(\lambda_1,\dots,\lambda_m)$ \cite{HL, GPPS}.
Probably, main results of the present paper can be generalized to the multiply connected case, but this is not done yet.

We emphasize that the $C^\infty$-smoothness of the boundary curve is essential for the validity of (\ref{1.13}). For example, the asymptotics is much more complicated for polygons \cite[Section 3]{GP}.

We return to considering a simply connected (probably multisheet) planar domain $\Omega$ bounded by a $C^\infty$-smooth closed curve.
Asymptotics (\ref{1.13}) allows us to introduce the {\it zeta function of the domain} $\Omega$
$$
\zeta_\Omega(s)=\mbox{\rm Tr}[\Lambda_\Omega^{-s}]=\sum\limits_{n=1}^\infty\big(\lambda_n(\Omega)\big)^{-s}.
$$
The series converges for $s\in{\mathbb C},\ \mbox{\rm Re}(s)>1$. In particular, $\zeta_{\mathbb D}=2\zeta_R$, where $\zeta_R(s)=\sum_{n=1}^\infty n^{-s}$ is the classical Riemann zeta function. Then $\zeta_\Omega$ extends to a meromorphic function on ${\mathbb C}$ with the unique simple pole at $s=1$. Moreover, the difference
$$
\zeta_\Omega(s)-2\left(\frac{L(\partial\Omega)}{2\pi}\right)^s\,\zeta_R(s)
$$
is an entire function \cite{E}.

Our main result is presented by the following two statements.

\begin{theorem} \label{Th1.1}
Let $\Omega$ be a simply connected (probably multisheet) planar domain bounded by a $C^\infty$-smooth closed curve $\partial\Omega$ and let $L(\partial\Omega)$ be the length of the curve $\partial\Omega$. The zeta function of the domain satisfies
$$
\zeta_\Omega(s)\geq2\left(\frac{L(\partial\Omega)}{2\pi}\right)^s\,\zeta_R(s)\quad\mbox{\rm for}\quad s\leq-1,
$$
where $\zeta_R$ is the classical Riemann zeta function.
\end{theorem}

\begin{theorem} \label{Th1.1'}
Let $\Omega$ be as in Theorem \ref{Th1.1}. If $\Omega$ is not a round disk, then
$$
{d\over ds}\Big[\Big(\frac{L(\partial\Omega)}{2\pi}\Big)^{s}\zeta_\Omega(-s)-2\zeta_R(-s)\Big]\geq C(\alpha)e^{\alpha s}\quad\mbox{\rm for}\quad s\geq1
$$
and
$$
\Big(\frac{L(\partial\Omega)}{2\pi}\Big)^{s}\zeta_\Omega(s)-2\zeta_R(-s)\geq C_1(\alpha)e^{\alpha s}\quad\mbox{\rm for}\quad s\geq1
$$
for any positive $\alpha$ with some positive constants $C(\alpha),C_1(\alpha)$.
\end{theorem}

Each of Theorems \ref{Th1.1} and \ref{Th1.1'} has two other equivalent forms that are also of some interest.

Let ${\mathbb S}=\partial{\mathbb D}=\{e^{i\theta}\}\subset{\mathbb C}$ be the unit circle. To simplify further formulas, the Dirichlet-to-Neumann operator of the unit disc will be denoted by $\Lambda:C^\infty({\mathbb S})\rightarrow C^\infty({\mathbb S})$, i.e., $\Lambda=\Lambda_{\mathbb D}$ (this operator was denoted by $\Lambda_e$ in \cite{JS} and \cite{MS}). Given a positive function $a\in C^\infty({\mathbb S})$, the first order pseudodifferential operator
$$
\Lambda_a=a^{1/2}\Lambda a^{1/2}=a^{-1/2}(a\Lambda)a^{1/2}:C^\infty({\mathbb S})\rightarrow C^\infty({\mathbb S})
$$
is non-negative and self-adjoint with respect to the $L^2$-product $(\cdot,\cdot)_{L^2({\mathbb S})}$. Here $a^{1/2}$ stands for the operator of multiplication by the function $a^{1/2}$. The eigenvalue spectrum of $\Lambda_a$
$$
\mbox{\rm Sp}(a)=\{0=\lambda_0(a)<\lambda_1(a)\leq\lambda_2(a)\leq\dots\}
$$
is called the {\it Steklov spectrum of the function} $a$ (or of the operator $\Lambda_a$). The spectrum $\mbox{\rm Sp}(a)$ has the same asymptotics
\begin{equation}
\lambda_n(a)=\frac{2\pi}{L(a)}\lambda^0_n+O(n^{-N})\quad\mbox{\rm as}\quad n\rightarrow\infty\quad\mbox{\rm for any}\quad N>0,
                                            \label{1.1}
\end{equation}
where $L(a)=\int_{\mathbb S}a^{-1}(\theta)\,d\theta$.
The {\it zeta function of} $a$ is defined by
\begin{equation}
\zeta_a(s)=\mbox{\rm Tr}[\Lambda_a^{-s}]=\sum\limits_{n=1}^\infty\big(\lambda_n(a)\big)^{-s}\quad \mbox{\rm for}\quad\mbox{\rm Re}\,s>1.
                                            \label{1.0}
\end{equation}
Two kinds of the Steklov spectrum are related as follows. Given a smooth simply connected planar domain $\Omega$, choose a biholomorphism $\Phi:{\mathbb D}\rightarrow\Omega$ and define the function $0<a\in C^\infty({\mathbb S})$ by $a(z)=|\Phi'(z)|^{-1}\ (z\in{\mathbb S})$. Let $\phi:{\mathbb S}\rightarrow\partial\Omega$ be the restriction of $\Phi$ to ${\mathbb S}$. Then $a\Lambda=\phi^*\Lambda_\Omega\,\phi^{*-1}$ and $\mbox{\rm Sp}(a)=\mbox{\rm Sp}(\Omega)$. See \cite[Section 3]{JS} for details. In particular, $\zeta_a=\zeta_\Omega$ is a meromorphic function on ${\mathbb C}$. Theorems \ref{Th1.1} and \ref{Th1.1'} are equivalent to the following statements.

\begin{theorem} \label{Th1.2}
Given a positive function $a\in C^\infty({\mathbb S})$, the inequality
\begin{equation}
\zeta_a(s)\geq2\left(\frac{L(a)}{2\pi}\right)^s\,\zeta_R(s)
                                                    \label{th1}
\end{equation}
holds for all $s\leq-1$, where $L(a)=\int_{-\pi}^\pi a^{-1}(e^{i\theta})\,d\theta$.
\end{theorem}

\begin{theorem} \label{Th1.2'}
For a positive function $a\in C^\infty({\mathbb S})$, the following alternative is valid. Either

(1) $\zeta_a(s)$ is identically equal to $2\left(\frac{L(a)}{2\pi}\right)^s\,\zeta_R(s)$

\noindent
or

(2) the difference
$
\Big(\frac{L(a)}{2\pi}\Big)^{-s}\zeta_a(s)-2\zeta_R(s)
$ satisfies
\begin{equation}
{d\over ds}\Big[\Big(\frac{L(a)}{2\pi}\Big)^{s}\zeta_a(-s)-2\zeta_R(-s)\Big]\geq C(\alpha)e^{\alpha s}\quad\mbox{\rm for}\quad s\geq1
                             \label{1.2}
\end{equation}
and
\begin{equation}
\Big(\frac{L(a)}{2\pi}\Big)^{s}\zeta_a(-s)-2\zeta_R(-s)\geq C_1(\alpha)e^{\alpha s}\quad\mbox{\rm for}\quad s\geq1
                             \label{1.3}
\end{equation}
for any positive $\alpha$ with some positive constants $C(\alpha),C_1(\alpha)$.
\end{theorem}

Theorem \ref{Th1.2} is actually a corollary of Theorem \ref{Th1.2'}. Nevertheless, we will first present the proof of Theorem \ref{Th1.2} and then we will show how Theorem \ref{Th1.2'} can be proved by some specification of the same arguments.

Given a Riemannian metric $g$ on the unit disc ${\mathbb D}$, let $\Delta_g$ be the Laplace -- Beltrami operator of the metric. The {\it Dirichlet-to-Neumann operator of the metric}
$$
\Lambda_g:C^\infty({\mathbb S})\rightarrow C^\infty({\mathbb S})
$$
is defined by $\Lambda_gf=\frac{\partial u}{\partial\nu}$, where $\nu$ is the unit outer normal to ${\mathbb S}$ with respect to the metric $g$ and $u$ is the solution to the Dirichlet problem
$$
\Delta_gu=0\quad\mbox{\rm in}\quad{\mathbb D},\quad u|_{\mathbb S}=f.
$$
The {\it Steklov spectrum of the metric} $g$ is again non-negative and discrete
$$
\mbox{\rm Sp}(\Delta_g)=\{0=\lambda_0(g)<\lambda_1(g)\leq\lambda_2(g)\leq\dots\}
$$
and the {\it zeta function of the metric}
$$
\zeta_g(s)=\mbox{\rm Tr}[\Lambda_g^{-s}]=\sum\limits_{n=1}^\infty\big(\lambda_n(g)\big)^{-s}\quad \mbox{\rm for}\quad\mbox{\rm Re}\,s>1
$$
again extends to a meromorphic function on ${\mathbb C}$. See \cite[Section 2]{JS} for details. Theorems \ref{Th1.1} and \ref{Th1.2} are equivalent to the following

\begin{theorem} \label{Th1.3}
Given a Riemannian metric $g$ on the unit disc ${\mathbb D}$, the inequality
$$
\zeta_g(s)\geq2\left(\frac{L({\mathbb S})}{2\pi}\right)^s\,\zeta_R(s)
$$
holds for all $s\leq-1$, where $L({\mathbb S})$ is the length of ${\mathbb S}$ in the metric $g$.
\end{theorem}

Theorems \ref{Th1.1'} and \ref{Th1.2'} are equivalent to the following

\begin{theorem} \label{Th1.3'}
For a Riemannian metric $g$ on the unit disc ${\mathbb D}$, the following alternative is valid. Either

(1) $g$ is conformally equivalent to the standard Euclidean metric $e$, i.e., there exist a diffeomorphism $\Phi:{\mathbb D}\rightarrow{\mathbb D}$ and positive function $\rho\in C^\infty({\mathbb D})$ such that $\rho|_{\mathbb S}=1$ and $g=\rho\Phi^*e$;

\noindent
or

(2) the difference
$
\Big(\frac{L({\mathbb S})}{2\pi}\Big)^{-s}\zeta_g(s)-2\zeta_R(s)
$ satisfies
$$
{d\over ds}\Big[\Big(\frac{L({\mathbb S})}{2\pi}\Big)^{s}\zeta_g(-s)-2\zeta_R(-s)\Big]\geq C(\alpha)e^{\alpha s}\quad\mbox{\rm for}\quad s\geq1
$$
and
$$
\Big(\frac{L({\mathbb S})}{2\pi}\Big)^{s}\zeta_g(-s)-2\zeta_R(-s)\geq C_1(\alpha)e^{\alpha s}\quad\mbox{\rm for}\quad s\geq1
$$
for any positive $\alpha$ with some positive constants $C(\alpha),C_1(\alpha)$.
\end{theorem}

We are grateful to G. Rozenblum for a discussion of some questions related to the paper.

\section{Proof of Theorems \ref{Th1.2} and \ref{Th1.2'}}

For a function $u$ on the unit circle ${\mathbb S}=\{e^{i\theta}\}$, we will write $u(\theta)$ instead of $u(e^{i\theta})$.
Introduce the operators
$$
D=-i\frac{d}{d\theta}:C^\infty({\mathbb S})\rightarrow C^\infty({\mathbb S}),\quad
\Lambda=\Lambda_{\mathbb D}=\Big(-\frac{d^2}{d\theta^2}\Big)^{1/2}:C^\infty({\mathbb S})\rightarrow C^\infty({\mathbb S}).
$$
Both $D$ and $\Lambda$ are self-adjoint operators with respect to the standard $L^2$-product
$$
(u,v)_{L^2}=\int\limits_{-\pi}^\pi u(\theta)\bar v(\theta)\,d\theta
$$
and satisfy
\begin{equation}
(\Lambda u,u)_{L^2}\geq\big|(Du,u)_{L^2}\big|
                             \label{2.0}
\end{equation}
for every $u\in C^\infty({\mathbb S})$. This follows from equalities $De^{in\theta}=n\,e^{in\theta}$ and $\Lambda e^{in\theta}=|n|\,e^{in\theta}$ for $n\in{\mathbb Z}$.

It suffices to prove Theorem \ref{Th1.2} for a function
$0<a\in C^\infty({\mathbb S})$ normalized by the condition
\begin{equation}
\frac{1}{2\pi}\int\limits_{-\pi}^\pi a^{-1}(\theta)\,d\theta=1.
                                             \label{2.1}
\end{equation}
This condition is always assumed in the current section.
Given such a function, we define the operators  $D_a,\Lambda_a:C^\infty({\mathbb S})\rightarrow C^\infty({\mathbb S})$ by
$$
D_a=a^{1/2}Da^{1/2},\quad  \Lambda_a=a^{1/2}\Lambda a^{1/2}.
$$

\begin{lemma} \label{L2.1}
For every integer $n$, the function
\begin{equation}
\varphi_n(\theta)=\frac{1}{\sqrt{2\pi a(\theta)}}\exp\Big(in\int\limits_0^\theta a^{-1}(s)\,ds\Big)
                         \label{2.4}
\end{equation}
is the eigenfunction of the operator $D_a$ associated to the eigenvalue $n$, i.e.
\begin{equation}
D_a\varphi_n=n\varphi_n.
                         \label{2.5}
\end{equation}
The inequality
\begin{equation}
(\Lambda_a\varphi_n,\varphi_n)_{L^2}\geq|n|
                         \label{2.6}
\end{equation}
holds for every integer $n$.
The family $\{\varphi_n\}_{n\in{\mathbb Z}}$ is an orthonormal basis of $L^2({\mathbb S})$.
\end{lemma}

\begin{proof}
First of all, being defined by (\ref{2.4}), $\varphi_n(\theta)$ is a $2\pi$-periodic function as is seen from (\ref{2.1}), i.e., $\varphi_n\in C^\infty({\mathbb S})$.
Equality (\ref{2.5}) is proved by a straightforward calculation on the base of definition (\ref{2.4}). It implies $(\varphi_n,\varphi_m)_{L^2}=0$ for $n\neq m$.
On using (\ref{2.1}), we also check
$\|\varphi_n\|^2_{L^2}=1$.
Thus, $\{\varphi_n\}_{n\in{\mathbb Z}}$ is an orthonormal system in $L^2({\mathbb S})$. Assume a function $u\in L^2({\mathbb S})$ to be orthogonal to all $\varphi_n$, i.e.,
$$
(u,\varphi_n)_{L^2}=\frac{1}{\sqrt{2\pi}}\int\limits_{-\pi}^\pi u(\theta)a^{-1/2}(\theta)\exp\Big(-in\int\limits_0^\theta a^{-1}(s)\,ds\Big)d\theta=0
\quad (n\in{\mathbb Z}).
$$
Change the integration variable in this equality by
$\alpha=\alpha(\theta)=\int_0^\theta a^{-1}(s)\,ds$ and introduce the functions $b\in C^\infty({\mathbb S})$ and $v\in L^2({\mathbb S})$ by $b(\alpha(\theta))=a(\theta)$ and $v(\alpha(\theta))=u(\theta)$. The previous formula takes the form
$$
\frac{1}{\sqrt{2\pi}}\int\limits_{-\pi}^\pi b^{1/2}(\alpha)v(\alpha)e^{-in\alpha}\,d\alpha=0
\quad (n\in{\mathbb Z}).
$$
This means that all Fourier coefficients of the function $b^{1/2}v$ are equal to zero. Hence $b^{1/2}v\equiv0$ and $u\equiv0$. We have thus proved $\{\varphi_n\}_{n\in{\mathbb Z}}$ is an orthonormal basis of $L^2({\mathbb S})$.

On using (\ref{2.0}) and (\ref{2.5}), we derive
$$
\begin{aligned}
(\Lambda_a\varphi_n,\varphi_n)_{L^2}&=(a^{1/2}\Lambda a^{1/2}\varphi_n,\varphi_n)_{L^2}=
(\Lambda (a^{1/2}\varphi_n),a^{1/2}\varphi_n)_{L^2}\\
&\geq\big|(D(a^{1/2}\varphi_n),a^{1/2}\varphi_n)_{L^2}\big|
=\big|(D_a\varphi_n,\varphi_n)_{L^2}\big|=|n|.
\end{aligned}
$$
\end{proof}

Both $D_a$ and $\Lambda_a$ are self-adjoint operators. Observe also that $D_a=a^{-1/2}(aD)a^{1/2}$ and $\Lambda_a=a^{-1/2}(a\Lambda)a^{1/2}$. In particular, the operators $D_a$ and $aD$ have coincident spectra as well as the operators $\Lambda_a$ and $a\Lambda$ have coincident spectra.
One easily compute
\begin{equation}
D^2_a=a^2D^2+2a(Da)D+\frac{1}{2}a(D^2a)+\frac{1}{4}(Da)^2.
                             \label{2.2}
\end{equation}

If $\xi$ is the Fourier dual variable for $\theta$, then the full symbol of $\Lambda$ is expressed in coordinates $(\theta,\xi)$ by the formula $\sigma_\Lambda(\theta,\xi)=|\xi|$.
Let us recall the classical formula for the full symbol of the product of two pseudodifferential operators
$$
\sigma_{AB}\sim \sum\limits_{\alpha=0}^\infty \frac{1}{\alpha!}(\partial^\alpha_\xi\sigma_A)D^\alpha\sigma_B.
$$
On using the formula, we compute
$$
\sigma_{\Lambda_a}=a|\xi|+\frac{1}{2}(Da)\,\mbox{sgn}(\xi)
$$
and
$$
\sigma_{\Lambda^2_a}=a^2\xi^2+2a(Da)\xi+\frac{1}{2}a(D^2a)+\frac{1}{4}(Da)^2.
$$
Comparing the last formula with (\ref{2.2}), we see that the operators $D_a^2$ and $\Lambda_a^2$ have coincident full symbols, i.e., the difference
$\Lambda_a^2-D_a^2$ is a smoothing operator.

The operator $D_a^2$ is a non-negative elliptic self-adjoint second order differential operator
with the same eigenvalue spectrum as $\Lambda^2$, i.e.,
$$
\mbox{\rm Sp}(D_a^2)=\{0,1,1,4,4,9,9,\dots\}.
$$
We will denote $|D_a|=(D_a^2)^{1/2}$. Observe that
\begin{equation}
\mbox{\rm Sp}(|D_a|)=\mbox{\rm Sp}(\Lambda)=\{0,1,1,2,2,3,3,\dots\}.
                      \label{2.7'}
\end{equation}

We will need complex powers of operators $|D_a|$ and $\Lambda_a$. There is a small difficulty in defining the powers since these operators are not invertible. As follows from Lemma \ref{L2.1}, null-spaces of $|D_a|$ and $\Lambda_a$ coincide and are equal to the one-dimensional space spanned by the function $\varphi_0$. Let $P_0$ be the orthogonal projection of $L^2({\mathbb S})$ onto that one-dimensional space. Then $|D_a|+P_0$ and $\Lambda_a+P_0$ are invertible elliptic first order differential operators. Therefore the powers $(|D_a|+P_0)^s$ and $(\Lambda_a+P_0)^s$ are well defined for every complex $s$, see \cite{Se} and \cite[Chapter 2]{Sh}.   $(|D_a|+P_0)^s$ and $(\Lambda_a+P_0)^s$ are pseudodifferential operator of order $\mbox{Re}\,s$. With some ambiguity, we will write
$$
|D_a|^s=(|D_a|+P_0)^s-P_0,\quad \Lambda_a^s=(\Lambda_a+P_0)^s-P_0.
$$
This actually coincides with Edward's definition \cite[Section 3]{E}. Observe that $|D_a|^s=(D_a^2)^{s/2}$. Equality (\ref{2.5}) implies
\begin{equation}
|D_a|^{s}\varphi_n=|n|^s\varphi_n\quad(s\in{\mathbb C}).
                         \label{2.8}
\end{equation}

\begin{lemma} \label{L2.2}
For every complex $s$,
$$
R(s)=\Lambda_a^s-|D_a|^{s}
$$
is a smoothing operator with finite trace and
\begin{equation}
\mbox{\rm Tr}[R(s)]=\zeta_a(-s)-2\zeta_R(-s)\quad(s\in{\mathbb C}).
                         \label{2.9}
\end{equation}
If $\{\chi_n\}_{n=0}^\infty$ is an orthonormal basis in $L^2({\mathbb S})$ and $\chi_n\in C^\infty({\mathbb S})\ (n=0,1,\dots)$, then
\begin{equation}
\sum\limits_{n=0}^\infty(R_s\chi_n,\chi_n)_{L^2}=\zeta_a(-s)-2\zeta_R(-s)\quad(s\in{\mathbb C}).
                         \label{2.10}
\end{equation}
\end{lemma}

\begin{proof}
For $\mbox{Re}\,s<-1$, equality (\ref{2.9}) follows from (\ref{1.0}) and (\ref{2.7'}):
\begin{equation}
\mbox{\rm Tr}[R(s)]=\mbox{\rm Tr}[\Lambda_a^s]-\mbox{\rm Tr}[|D_a|^s]=\sum\limits_{n=1}^\infty\big(\lambda_n(a)\big)^{s}-2\sum\limits_{n=1}^\infty n^s=\zeta_a(-s)-2\zeta_R(-s)
\quad(\mbox{Re}\,s<-1).
                         \label{2.10'}
\end{equation}

Let us recall some general facts. Let $A$ be an elliptic first order pseudodifferential operator on a manifold $M$. Moreover, assume $A$ to be self-adjoint and positive with respect to the $L^2$-product defined with the help of a smooth volume form on $M$. The positiveness means $(Au,u)_{L^2}\geq\delta\|u\|^2_{L^2}$ for every $u\in C^\infty_0(M)$ with some $\delta>0$ independent of $u$. Then the power $A^s$ is a well defined pseudodifferential operator for every $s\in{\mathbb C}$. Let $K_s(x,y)\ (x,y\in M)$ be the Schwartz kernel of $A^s$. By statement (iii) of \cite[Theorem 4]{Se}, $K_s(x,x)$ is a meromorphic function of the argument $s\in{\mathbb C}$ whose poles and residues are uniquely determined by the full symbols of $A$.

Returning to our case, the operators $\Lambda_a+P_0$ and $|D_a|+P_0$ have the same full symbol. Therefore the Schwartz kernel $K_s(\theta,\theta)$ of
$R(s)=(\Lambda_a+P_0)^s-(|D_a|+P_0)^{s}$ is an entire function and $\mbox{\rm Tr}[R(s)]=\int_{\mathbb S}K_s(\theta,\theta)\,d\theta$ is also an entire function. Together with (\ref{2.10'}), this proves the first statement of the lemma.
The second statement is just the standard property of compact finite trace self-adjoint operators on a Hilbert space \cite[Chapter 3]{Si}.
\end{proof}

\begin{lemma} \label{L2.3}
Let a function $\varphi\in C^\infty({\mathbb S})$ satisfy $\|\varphi\|_{L^2}=1$. Then, for every $s\geq1$,
\begin{equation}
(\Lambda_a^s\varphi,\varphi)_{L^2}\geq(\Lambda_a\varphi,\varphi)^s_{L^2}.
                         \label{2.21}
\end{equation}
\end{lemma}

\begin{proof}
There exists an orthonormal basis $\{\psi_n\}_{n=0}^\infty$ in $L^2({\mathbb S})$ consisting of eigenfunctions of $\Lambda_a$, i.e.,
$\Lambda_a\psi_n=\lambda_n(a)\psi_n$. Expand the function $\varphi$ in the basis
$$
\varphi=\sum\limits_{n=0}^\infty\alpha_n\psi_n,\quad \sum\limits_{n=0}^\infty|\alpha_n|^2=1.
$$
Then
$$
(\Lambda_a\varphi,\varphi)_{L^2}=\sum\limits_{n=0}^\infty|\alpha_n|^2\lambda_n(a).
$$
Since $x\mapsto x^s$ is the convex function on $[0,\infty)$ for $s\geq1$, by the Young inequality,
\begin{equation}
(\Lambda_a\varphi,\varphi)^s_{L^2}\leq\sum\limits_{n=0}^\infty|\alpha_n|^2\big(\lambda_n(a)\big)^s=(\Lambda^s_a\varphi,\varphi)_{L^2}.\label{a1}
\end{equation}
\end{proof}

\begin{proof}[Proof of Theorem \ref{Th1.2}.]
By Lemmas \ref{L2.2} and \ref{L2.1},
$$
\zeta_a(-s)-2\zeta_R(-s)=\mbox{\rm Tr}[R_s]=\sum\limits_{n=-\infty}^\infty\big((\Lambda_a^s\varphi_n,\varphi_n)_{L^2}-
(|D_a|^s\varphi_n,\varphi_n)_{L^2}\big).
$$
Each term of the latter series is non-negative. Indeed, by Lemma \ref{L2.3}, (\ref{2.6}), and (\ref{2.8}), for $s\geq1$,
\begin{equation}
(\Lambda_a^s\varphi_n,\varphi_n)_{L^2}-(|D_a|^s\varphi_n,\varphi_n)_{L^2}\geq
(\Lambda_a\varphi_n,\varphi_n)^s_{L^2}-(|D_a|^s\varphi_n,\varphi_n)_{L^2}\geq|n|^s-|n|^s=0.
                                                                     \label{a2}
\end{equation}
This proves \eqref{th1}.
\end{proof}

In the rest of the section, we study the question: When can we get the strong inequality in (\ref{th1})? To this end we recall the following definition from \cite{JS, MS}.

Two functions $a,b\in C^\infty({\mathbb S})$ are said to be {\it conformally equivalent} if there exists a conformal or anticonformal transformation $\Psi$ of the disc ${\mathbb D}$ such that
$$
b=a\circ\psi\left|\frac{d\psi}{d\theta}\right|^{-1},\quad\mbox{\rm where}\quad\psi=\Psi|_{\mathbb S}.
$$
($\Psi$ is anticonformal if $\bar\Psi$ is conformal.)

If two positive functions $a,b\in C^\infty({\mathbb S})$ are conformally equivalent, then $\mbox{Sp}(a)=\mbox{Sp}(b)$ and $\zeta_a\equiv\zeta_b$. In particular, if $0<a\in C^\infty({\mathbb S})$ is conformally equivalent to $\mathbf 1$ (the function identically equal to 1), then $\zeta_a\equiv2\zeta_R$.

\begin{lemma}            \label{L2.5}
Let a positive function $a\in C^\infty({\mathbb S})$ be normalized by condition (\ref{2.1}). For the function $\varphi_1$ defined by (\ref{2.4}), the equality
\begin{equation}
(\Lambda_a\varphi_1,\varphi_1)_{L^2}=1
                                                        \label{2.26}
\end{equation}
holds if and only if $a$ is conformally equivalent to $\mathbf 1$.
\end{lemma}

\begin{proof}
We prove the ``if'' statement. Let $\Psi$ be a conformal or anticonformal transformation of the disk ${\mathbb D}$ and let a real function $b(\theta)$ be such that
\begin{equation}
\Psi(e^{i\theta})=e^{ib(\theta)}.
                                                        \label{2.27}
\end{equation}
Assume $a$ to be conformally equivalent to $\mathbf 1$ via $\Psi$, i.e.,
\begin{equation}
a(\theta)=\Big|\frac{db(\theta)}{d\theta}\Big|^{-1}.
                                                        \label{2.28}
\end{equation}
The derivative $db/d\theta$ is positive in the case of a conformal $\Psi$ and $db/d\theta$ is negative in the case of an anticonformal $\Psi$.

In the case of a conformal $\Psi$, the function $\Psi(e^{i\theta})$ is the restriction to ${\mathbb S}$ of the holomorphic on ${\mathbb D}$ function $\Psi(z)$. Therefore
\begin{equation}
\Lambda\Psi(e^{i\theta})=D\Psi(e^{i\theta})=a^{-1}(\theta)\Psi(e^{i\theta}).
                                                        \label{2.29}
\end{equation}
From definition (\ref{2.4}), we derive with the help of (\ref{2.27})--(\ref{2.28})
$$
a^{1/2}\varphi_1(\theta)=\frac{e^{i\theta_0}}{\sqrt{2\pi}}\Psi(e^{i\theta})
$$
with some $\theta_0\in{\mathbb R}$. This gives together with (\ref{2.29})
$$
\begin{aligned}
(\Lambda_a\varphi_1,\varphi_1)_{L^2}&=\big(\Lambda(a^{1/2}\varphi_1),a^{1/2}\varphi_1\big)_{L^2}=\frac{1}{2\pi}\big(\Lambda\Psi(e^{i\theta}),\Psi(e^{i\theta})\big)_{L^2}\\
&=\frac{1}{2\pi}\big(a^{-1}(\theta)\Psi(e^{i\theta}),\Psi(e^{i\theta})\big)_{L^2}=\frac{1}{2\pi}\int\limits_0^{2\pi}a^{-1}(\theta)\,d\theta=1.
\end{aligned}
$$
We proceed similarly in the case of an anticonformal $\Psi$.

We now prove the ``only if'' statement. Assume (\ref{2.26}) to be valid. Set
\begin{equation}
b(\theta)=\int_0^\theta a^{-1}(s)\,ds\quad\mbox{and}\quad \psi=e^{ib}=\sqrt{2\pi}\,a^{1/2}\varphi_1.
                                                        \label{2.30}
\end{equation}
Observe that $\psi:{\mathbb S}\rightarrow{\mathbb S}$ is an orientation preserving diffeomorphism. Since
$$
(\Lambda\psi,\psi)_{L^2}=2\pi(a^{1/2}\Lambda a^{1/2}\varphi_1,\varphi_1)_{L^2}=2\pi(\Lambda_a\varphi_1,\varphi_1)_{L^2}=2\pi
$$
and
$$
(D\psi,\psi)_{L^2}=2\pi(a^{1/2}D a^{1/2}\varphi_1,\varphi_1)_{L^2}=2\pi(D_a\varphi_1,\varphi_1)_{L^2}=2\pi,
$$
we have
\begin{equation}
(\Lambda\psi,\psi)_{L^2}=(D\psi,\psi)_{L^2}.
                                                        \label{2.31}
\end{equation}
Both sides of (\ref{2.31}) are expressed in terms of Fourier coefficients ${\hat\psi}_k$ of the function $\psi$ as follows:
$$
(\Lambda\psi,\psi)_{L^2}=\sum\limits_{k=-\infty}^\infty |k|\,|{\hat\psi}_k|^2,\quad
(D\psi,\psi)_{L^2}=\sum\limits_{k=-\infty}^\infty k\,|{\hat\psi}_k|^2.
$$
Equality (\ref{2.31}) holds only when ${\hat\psi}_k=0$ for all $k<0$.

We define the holomorphic function on $\D$ by
$$
\Psi(z)=\sum\limits_{k=0}^\infty {\hat\psi}_k\,z^k.
$$
The map $\Psi|_{\mathbb S}=\psi:{\mathbb S}\rightarrow{\mathbb S}$ is an orientation preserving diffeomorphism. By the argument principle, $\Psi$ is a conformal transformation of the disk ${\mathbb D}$. Our definitions of $b$ and $\Psi$ imply the validity of equalities (\ref{2.27})--(\ref{2.28}) that mean the conformal equivalence of the functions $a$ and $\mathbf 1$.
\end{proof}

\begin{lemma}     \label{L2.6}
For a positive function $a\in C^\infty({\mathbb S})$ satisfying (\ref{2.1}) and for functions $\varphi_n$ defined by (\ref{2.4}), the following three statements are equivalent:

{\rm (i)} $(\Lambda_a\varphi_1,\varphi_1)_{L^2}=1 $;

{\rm (ii)} $(\Lambda_a\varphi_n,\varphi_n)_{L^2}=n$ and $(\Lambda_a\varphi_{n+1},\varphi_{n+1})_{L^2}=n+1$ for some $2\leq n\in{\mathbb N}$;

{\rm (iii)} $(\Lambda_a\varphi_n,\varphi_n)_{L^2}=|n|$ for every $n\in \Z$.
\end{lemma}

\begin{proof}
The fact that (i) implies (iii) follows from Lemma \ref{L2.5}. Obviously (iii) implies (ii).
It remains to prove that (ii) implies (i).

Define the diffeomorphism $\psi:{\mathbb S}\rightarrow{\mathbb S}$ by (\ref{2.30}). Assuming statement (ii) to be valid and repeating our arguments from the proof of Lemma \ref{L2.5}, we see that, for some $2\leq n\in{\mathbb N}$, the maps $\psi^n$ and $\psi^{n+1}$ extend to some holomorphic maps $\Psi^{(n)}:{\mathbb D}\rightarrow{\mathbb D}$ and
$\Psi^{(n+1)}:{\mathbb D}\rightarrow{\mathbb D}$ respectively.

Then $\Psi=\Psi^{(n+1)}/\Psi^{(n)}$ is a meromorphic function with the boundary trace
$\Psi|_{\mathbb S}=\psi$. Let us demonstrate that actually $\Psi$ is a holomorphic function. Indeed, $\Psi^n$ is a meromorphic function with the boundary trace
$\Psi^n|_{\mathbb S}=\psi^n=\Psi^{(n)}|_{\mathbb S}$. Therefore the difference $\Psi^n-\Psi^{(n)}$ is a meromorphic function with the zero boundary trace. The difference has finitely many poles that belong to the interior of the disk ${\mathbb D}$. Therefore there exists a polynomial $P(z)$ not identically equal to zero such that the product $P(\Psi^n-\Psi^{(n)})$ is a holomorphic function on the disc. Since the product has the zero boundary trace, it must be identically equal to zero, i.e., $\Psi^n=\Psi^{(n)}$. Thus, $\Psi^n$ is a holomorphic function and $\Psi$ is also a holomorphic function with the boundary trace $\Psi|_{\mathbb S}=\psi$. Again, by the argument principle, $\Psi$ is a conformal transformation of the disc and we finish the proof as in the proof of Lemma \ref{L2.5}.
\end{proof}

\begin{proof}[Proof of Theorem \ref{Th1.2'}.]
Let $a\in C^\infty({\mathbb S})$ be a positive function satisfying (\ref{2.1}).
Let $1\le t\le s$. On using the convexity of the function $x\mapsto x^{s/t}$ on $[0,\infty)$, we obtain similarly to \eqref{2.21}
$$
(\Lambda_a^t\varphi_n,\varphi_n)^{s/t}_{L^2}\leq (\Lambda^s_a\varphi_n,\varphi_n)_{L^2}
$$
for every $n\in{\mathbb Z}$.
On the other hand, (\ref{2.8}) and \eqref{a2} with $t$ in place of $s$ imply
$$
(\Lambda_a^t\varphi_n,\varphi_n)^{s/t}_{L^2}=(\Lambda_a^t\varphi_n,\varphi_n)^{{s/t}-1}_{L^2}(\Lambda_a^t\varphi_n,\varphi_n)_{L^2}\ge |n|^{s-t}(\Lambda_a^t\varphi_n,\varphi_n)_{L^2}.
$$
From two last inequalities
\begin{equation}
(\Lambda^s_a\varphi_n,\varphi_n)_{L^2}\geq |n|^{s-t}(\Lambda_a^t\varphi_n,\varphi_n)_{L^2}\quad (1\le t\le s,\ n\in{\mathbb N}).
                         \label{2.22}
\end{equation}
We recall that \eqref{2.6} and  (\ref{2.21}) already give
\begin{equation}
(\Lambda^t_a\varphi_n,\varphi_n)_{L^2}\geq |n|^t\quad (t\geq 1,\ n\in{\mathbb N}).
                         \label{2.23}
\end{equation}

On using (\ref{2.22})--(\ref{2.23}) and the equality $\Lambda_a\varphi_0=0$, we derive for $1\le t\le s$
\begin{equation}
\begin{aligned}
\zeta_a(-s)-2\zeta_R(-s)&=\sum\limits_{n=-\infty}^\infty\big((\Lambda_a^s\varphi_n,\varphi_n)_{L^2}-|n|^s\big)
\geq\sum\limits_{n=-\infty}^\infty|n|^{s-t}\big((\Lambda_a^t\varphi_n,\varphi_n)_{L^2}-|n|^t\big)\\
&\geq\sum\limits_{n=-\infty}^\infty\big((\Lambda_a^t\varphi_n,\varphi_n)_{L^2}-|n|^t\big)=\zeta_a(-t)-2\zeta_R(-t).
\end{aligned}
                         \label{2.24}
\end{equation}

{\bf Remark.}
Actually our arguments give the stronger version of \eqref{2.24}
$$
\zeta_a(-s)-2\zeta_R(-s)\ge \T\big[|D_a|^{s-t}(\Lambda_a^t-|D_a|^t)\big]\ge\zeta_a(-t)-2\zeta_R(-t)\quad (1\leq t\leq s).
$$

If $a$ is conformally equivalent to $\mathbf 1$, then $\zeta_a\equiv 2\zeta_R$.

Now, assume $a$ is not conformally equivalent to $\mathbf 1$. By Lemma \ref{L2.6}, for every $N>0$, there exists $2\leq n_0\in{\mathbb N}$ such that $n_0\geq N$ and
 $(\Lambda_a\varphi_{n_0},\varphi_{n_0})_{L^2}>n_0$.
Then $(\Lambda_a^t\varphi_{n_0},\varphi_{n_0})_{L^2}>|n_0|^t$ for every $t\geq1$ as follows from (\ref{2.22}). All terms of the first sum in (\ref{2.24}) are non-negative and at least one term is positive.

Moreover, we can prove the positiveness of the derivative. To simplify our formulas, let us introduce the notation
$\psi(t)=\zeta_a(-t)-2\zeta_R(-t)$ and set $s=t+\Delta t$ with $\Delta t>0$. We rewrite (\ref{2.24}) in more detail as follows:
$$
\begin{aligned}
\psi(t+\Delta t)&\geq
n_0^{\Delta t}\big((\Lambda_a^t\varphi_{n_0},\varphi_{n_0})_{L^2}-n_0^t\big)+\sum\limits_{n\neq n_0}|n|^{\Delta t}\big((\Lambda_a^t\varphi_n,\varphi_n)_{L^2}-|n|^t\big)\\
&\geq
n_0^{\Delta t}\big((\Lambda_a^t\varphi_{n_0},\varphi_{n_0})_{L^2}-n_0^t\big)+\sum\limits_{n\neq n_0}\big((\Lambda_a^t\varphi_n,\varphi_n)_{L^2}-|n|^t\big)\\
&=(n_0^{\Delta t}-1)\big((\Lambda_a^t\varphi_{n_0},\varphi_{n_0})_{L^2}-n_0^t\big)+\sum\limits_{n=-\infty}^\infty\big((\Lambda_a^t\varphi_n,\varphi_n)_{L^2}-|n|^t\big)\\
&=(n_0^{\Delta t}-1)\big((\Lambda_a^t\varphi_{n_0},\varphi_{n_0})_{L^2}-n_0^t\big)+\psi(t).
\end{aligned}
$$
The result can be written as
$$
\psi(t+\Delta t)-\psi(t)\geq C_t(n_0^{\Delta t}-1)
$$
with
\begin{equation}
C_t=(\Lambda_a^t\varphi_{n_0},\varphi_{n_0})_{L^2}-n_0^t.
                         \label{2.40}
\end{equation}
This implies
\begin{equation}
\frac{d\psi(t)}{dt}\geq C_t\ln n_0 >0.
                         \label{2.41}
\end{equation}

The coefficient $C_t$ in (\ref{2.41}) grows exponentially with $t$. Indeed, on using (\ref{2.6}) and (\ref{2.21}), we obtain from (\ref{2.40}) for $t\geq1$
$$
\begin{aligned}
C_t&=(\Lambda_a^t\varphi_{n_0},\varphi_{n_0})_{L^2}-n_0^t\geq (\Lambda_a\varphi_{n_0},\varphi_{n_0})^t_{L^2}-n_0^t
=(\Lambda_a\varphi_{n_0},\varphi_{n_0})^{t-1}_{L^2}(\Lambda_a\varphi_{n_0},\varphi_{n_0})_{L^2}-n_0^t\\
&\geq n_0^{t-1}(\Lambda_a\varphi_{n_0},\varphi_{n_0})_{L^2}-n_0^t=n_0^{t-1}\Big((\Lambda_a\varphi_{n_0},\varphi_{n_0})_{L^2}-n_0\Big),
\end{aligned}
$$
i.e.,
\begin{equation}
C_t\geq n_0^{t-1}\Big((\Lambda_a\varphi_{n_0},\varphi_{n_0})_{L^2}-n_0\Big)=(\ln n_0)^{-1}Ce^{\alpha t},
                         \label{2.42}
\end{equation}
where $\alpha=\ln n_0>0$ and $C=n_0^{-1}(\ln n_0)\Big((\Lambda_a\varphi_{n_0},\varphi_{n_0})_{L^2}-n_0\Big)>0$. From (\ref{2.41}) and (\ref{2.42}),
$$
\frac{d\psi(t)}{dt}\geq C e^{\alpha t}\quad\mbox{for}\quad t\geq0.
$$
This proves (\ref{1.2}) since $n_0$ can be chosen such that $\alpha=\ln n_0\geq\alpha_0$ for every $\alpha_0>0$. Estimate (\ref{1.3}) is obtained from (\ref{1.2}) by integration with the inequality $\psi(1)\ge 0$ taken into account.
\end{proof}

\section{some open questions}

Besides statements of Theorems \ref{Th1.2} and \ref{Th1.2'}, some other properties of the function $\zeta_a(s)$ can be studied. In particular, the behavior of the function on the interval $s\in(-1,0)$ is of a great interest. Observe that $\zeta_a(0)=2\zeta_R(0)=-1$ by Lemma \ref{L2.2}. The best way to discover interesting properties of the zeta functions is just to give a look at graphs of $\zeta_a(s)$ for several examples of the function $a$. Unfortunately, besides the trivial case of $a\equiv1$, no example is known with an explicit expression for $\zeta_a(s)$. Nevertheless, as mentioned in \cite[Section 3.2]{GP}, there exists an algorithm for computing Steklov eigenvalues of a planar domain. We hope the algorithm will allow us to compute the zeta function with a good precision. This work is not started yet.

The values
$$
Z_k(a)=\zeta_a(-2k)\quad (k=1,2,\dots)
$$
are of a specific interest. These values are called {\it zeta-invariants} of the function $a$. Let us recall that $\zeta_R(-2k)=0\ (k=1,2,\dots)$. Therefore Theorem \ref{Th1.2} implies the inequality
\begin{equation}
Z_k(a)\geq0\quad (k=1,2,\dots)
                         \label{3.1}
\end{equation}
for a positive function $a\in C^\infty({\mathbb S})$ (condition (\ref{2.1}) is not assumed now).

Zeta-invariants can be explicitly expressed through Fourier coefficients of the function $a$. Given $a\in C^\infty({\mathbb S})$, we denote its Fourier coefficients by ${\hat a}_n$, i.e.,
$$
a(\theta)=\sum\limits_{n=-\infty}^\infty {\hat a}_n\,e^{in\theta}.
$$
Then
\begin{equation}
Z_k(a)=\sum\limits_{j_1+\dots +j_{2k}=0} N_{j_1\dots j_{2k}}\,{\hat a}_{j_1}{\hat a}_{j_2}\dots {\hat a}_{j_{2k}},
                                               \label{3.2}
\end{equation}
where, for $j_1+\dots+ j_{2k}=0$,
\begin{equation}
\begin{aligned}
N_{j_1\dots j_{2k}}=\sum\limits_{n=-\infty}^\infty
\Big[&\left|n(n+j_1)(n+j_1+j_2)\dots(n+j_1+\dots+ j_{2k-1})\right|\\
&-n(n+j_1)(n+j_1+j_2)\dots(n+j_1+\dots +j_{2k-1})\Big].
\end{aligned}
                                               \label{3.3}
\end{equation}
There is only a finite number of nonzero summands on the right-hand side of (\ref{3.3}) since the expression
$$
f(n)=n(n+j_1)(n+j_1+j_2)\dots(n+j_1+\dots+ j_{2k-1})
$$
is a polynomial of degree $2k$ in $n$ which takes positive values for sufficiently large $|n|$.
Series (\ref{3.2}) converges absolutely since Fourier coefficients ${\hat a}_n$ decay rapidly while coefficients $N_{j_1\dots j_{2k}}$ are of a polynomial growth in $|j|=|j_1|+\dots+|j_{2k}|$.
In particular, for a real function $a\in C^\infty({\mathbb S})$,
\begin{equation}
Z_1(a)=\frac{2}{3}\sum\limits_{n=2}^\infty(n^3-n)\, |{\hat a}_n|^2.
                                               \label{3.4}
\end{equation}
This formula belongs to Edward \cite{E}. He also proved (\ref{3.2}) in the case of $k=2$ (without using the notation $Z_2(a)$) \cite{E2}. In the general case, zeta-invariants were introduced in
\cite[Section 2]{MS}.

We emphasize that formulas (\ref{3.2})--(\ref{3.3}) make sense for an arbitrary (complex-valued) function $a\in C^\infty({\mathbb S})$. These formulas can be taken as the definition of zeta-invariants for an arbitrary $a\in C^\infty({\mathbb S})$ although the zeta function $\zeta_a$ is not defined in the general case. Moreover, one can easily see that zeta invariants are real for a real function $a$.

\begin{conjecture} \label{C3.1}
Inequalities (\ref{3.1}) hold for every real function $a\in C^\infty({\mathbb S})$.
\end{conjecture}

The conjecture is true in the case of $k=1$ by Edward's formula (\ref{3.4}). For $k\geq2$, the conjecture remains open although it is confirmed by a lot of numerical experiments. Unlike the problem of computing Steklov eigenvalues, formulas (\ref{3.2})--(\ref{3.3}) are very easy for computerization.

Conjecture \ref{C3.1} can be strengthened by some estimate from below. For example, as follows from (\ref{3.4}),
\begin{equation}
Z_1(a)\geq c_1\sum\limits_{n=2}^\infty n^3|{\hat a}_n|^2
                                               \label{3.5}
\end{equation}
for every real $a\in C^\infty({\mathbb S})$ with some universal constant $c_1>0$. Observe that the Fourier coefficients ${\hat a}_0$ and ${\hat a}_1$ do not participate on the right-hand side of the estimate. This relates to the conformal invariance of zeta-invariants \cite[Section 4]{MS}. In particular, $Z_k(a)=0\ (k=1,2,\dots)$ for all functions belonging to the three-dimensional space
$$
\{a\in C^\infty({\mathbb S})\mid a(\theta)={\hat a}_0+{\hat a}_1e^{i\theta}+{\hat a}_{-1}e^{-i\theta}\}.
$$
The natural generalization of (\ref{3.5}) looks as follows:
\begin{equation}
Z_k(a)\geq c_k\sum\limits_{n=2}^\infty n^{2k+1}|{\hat a}_n|^{2k}\quad (k=1,2,\dots)
                                               \label{3.6}
\end{equation}
for every real $a\in C^\infty({\mathbb S})$ with some constant $c_k>0$ depending only on $k$. This estimate is conjectured in \cite[Problem 6.1]{MS}, it is not proved yet. In our opinion, the best possible version of such estimates should look as follows:
\begin{equation}
Z_k(a)\geq c_k\sum\limits_{n=2}^\infty n^{2k+1}|{\hat b}_n|^{2},\quad\mbox{\rm where}\quad b=a^k.
                                               \label{3.7}
\end{equation}

In \cite{E2}, Edward proved the pre-compactness of a Steklov isospectral family of planar domains in the $H^s$-topology for $s<5/2$. The proof is based on the usage of first two zeta-invariants
$Z_1(a),Z_2(a)$, and of $\zeta_a(-1)$, $\zeta_a(-3)$. The same approach would work for proving the corresponding compactness theorem in the $C^\infty$-topology if estimate (\ref{3.7}) was proven. But estimate (\ref{3.6}) is not sufficient for such a proof.

\end{document}